\documentclass[12pt,reqno]{amsart}
\usepackage[latin1]{inputenc}
\usepackage[T1]{fontenc}
\usepackage{mathptmx}
\usepackage{amssymb}
\usepackage{extarrows}
\usepackage{tikz}
\usepackage{soul}
\usetikzlibrary{positioning, calc}
\usetikzlibrary{arrows,decorations.markings}
 
\usepackage{color}

\newcommand{\ket}[1]{\left|#1\right\rangle}

\usepackage{paralist}

\usepackage{mathtools}
\DeclarePairedDelimiterX{\norm}[1]{\lVert}{\rVert}{#1}
\newtheorem{theorem}{Theorem}[section]
\newtheorem{lemma}[theorem]{Lemma}
\newtheorem{corollary}[theorem]{Corollary}
\newtheorem{proposition}[theorem]{Proposition}

\theoremstyle{definition}

\newtheorem{definition}[theorem]{Definition}

\numberwithin{equation}{section}

\begin{document}

\title[Optimal approximation to unitary quantum operators with linear optics]{Optimal approximation to unitary quantum operators with linear optics}

\author{Juan Carlos Garcia-Escartin}

\address{Departamento de Teor\'ia de la Se\~{n}al y Comunicaciones e Ingenier\'ia Telem\'atica. ETSI de Telecomunicaci\'on. Universidad de Valladolid. Campus Miguel Delibes. Paseo Bel\'en 15. 47011 Valladolid. Spain.}\email{juagar@tel.uva.es}

\author{Vicent Gimeno}

\address{Universitat Jaume I, Campus de Riu Sec, Departament de Matem\`atiques \& Institut Universitari de Matem\`atiques i Aplicacions de Castell\'o--IMAC, 12071
Caste\-ll\'on de la Plana, Spain} \email{gimenov@uji.es}

\author{Julio Jos\'e Moyano-Fern\'andez}

\address{Universitat Jaume I, Campus de Riu Sec, Departamento de Matem\'aticas \& Institut Universitari de Matem\`atiques i Aplicacions de Castell\'o--IMAC, 12071
Caste\-ll\'on de la Plana, Spain} \email{moyano@uji.es}

\thanks{The first author has been funded by Junta de Castilla y Le\'on (project VA296P18). The second author has been partially supported by the Research Program of the University Jaume I--Project UJI-B2018-3, as well as by the Spanish Government Ministerio de Econom\'ia y Competitividad (MINECO-FEDER) grant MTM2017-84851-C2-2. The third author was partially supported by the Spanish Government, Ministerios de Ciencia e Innovaci\'on y de Universidades, grant PGC2018-096446-B-C22, as well as by Universitat Jaume I, grant UJI-B2018-10}

\begin{abstract}
Linear optical systems acting on photon number states produce many interesting evolutions, but cannot give all the allowed quantum operations on the input state. Using Toponogov's theorem from differential geometry, we propose an iterative method that, for any arbitrary quantum operator $U$ acting on $n$ photons in $m$ modes, returns an operator $\widetilde{U}$ which can be implemented with linear optics. The approximation method is locally optimal and converges. The resulting operator $\widetilde{U}$ can be translated into an experimental optical setup using previous results.
\end{abstract}

\maketitle

\section{Introduction}
Linear optical devices under quantum light show a rich behaviour and have different applications in experiments on the foundations of quantum optics and quantum information \cite{Lou00,GR04,Leo10}. While they can be built with relatively simple optical elements like beam splitters and phase shifters \cite{RZB94,BA14,Saw16,CHM16,GMS18}, their behaviour for photon number states cannot be accurately reproduced by any classical system. One clear example is the boson sampling problem, for which quantum systems can give efficient solutions which cannot be produced by any classical method \cite{AA11}.

Classically, the evolution of the electrical field in $m$ orthogonal modes going through a linear optical system is perfectly described by a unitary $m\times m$ matrix, $S$, called the \emph{scattering matrix} of the system \cite{Poz04}. The evolution of $n$ photons distributed through these $m$ possible modes is given by an $M\times M$ unitary \emph{evolution matrix} $U$ acting on the $M={m+n-1 \choose n}$ states of the resulting Hilbert space. 

The \emph{photonic homomorphism} $\varphi: U(m) \to U(M)$ gives the evolution matrix $U$ which corresponds to a scattering matrix $S$ describing the linear optical system. $U=\varphi(S)$ can be computed from different equivalent methods \cite{Cai53,Sch04,SGL04}.

Any unitary matrix can be written as an exponential $U=e^{iH}$ for a Hermitian matrix $H$. In linear optical devices, we will call this matrix the \emph{effective Hamiltonian} $H_U$ of the linear system. Similarly, $S=e^{iH_S}$. 

The image of the photonic homomorphism, $\mathrm{im}(\varphi)$, is a subgroup of $U(M)$ which contains all the quantum evolutions that are allowed for $n$ photons a linear optical system with $m$ modes. The image subgroup is a representation of $U(m)$ in $U(M)$ and maps each possible classical scattering matrix $S$ describing a linear system into the quantum evolution $U=\varphi(S)$ it induces for $n$ photons.

The evolution in the corresponding unitary algebras from $iH_S$ to $iH_U$ is given by the differential of $\varphi$, $d \varphi:\mathfrak{u}(m)\to \mathfrak{u}(M)$, for which there are also explicit expressions \cite{LN04,FX97,AC05,BL81,ALN06,GGM18}. 

From the point of view of system design, a natural question is whether any given quantum evolution $U \in U(M)$ can be realized using only linear optics. From a simple dimensional argument, it is clear that, except when $m=1$ or $n=1$, there must be some impossible operations \cite{MG17}.

In a previous work, we have given an explicit inverse method to find the $S$ corresponding to any $U\in \mathrm{im}(\varphi)$ which can be implemented using linear optics \cite{GGM19}. 

Here, we address the problem of approximating $U\not\in \mathrm{im}(\varphi)$. We give a method to find the linear optics system with an evolution matrix $\widetilde{U}\in \mathrm{im}(\varphi)$ which minimizes the distance to $U$ locally. The result is based on Toponogov's comparison theorem \cite{CE08} from differential geometry.

Section \ref{imalg} describes the structure of the image algebra and its complement and states two theorems that will become useful later. Section \ref{notation} introduces the basic concepts from differential geometry used in the proof and the notation for the rest of the paper. Section \ref{section:bi} defines the bi-invariant Riemannian metric in which the results are given. Section \ref{section:topo} shows how to apply Toponogov's theorem to reduce the problem of approximating a unitary to finding a geodesic in the correct manifold. Section \ref{tricks} discusses some tricks related to the generation of random unitaries which are needed to explore the image group and to be able to compute a valid matrix logarithm for any desired $U$. Section \ref{algo} describes the iterative method that produces the desired approximation and gives some examples. Finally, Section \ref{summary}, gives a general overview of the method and comments on some practical problems and possible improvements for the approximation algorithm.

\section{The image algebra and its orthogonal complement}
\label{imalg}
If we study the induced map $d \varphi:\mathfrak{u}(m)\to \mathfrak{u}(M)$, we can decompose the Lie algebra $\mathfrak{u}(M)$ orthogonally so that 
\begin{equation}
\mathfrak{u}(M)={\rm im}\, d\varphi\oplus ({\rm im}\, d\varphi)^\perp,
\end{equation}
where $({\rm im}\, d\varphi)^\perp$ is the orthogonal complement of ${\rm im}\, d\varphi$ with respect to the metric
\begin{equation}
\label{metric}
\langle u,v\rangle=\frac{1}{2}{\rm tr}(u^\dag v+v^\dag u).
\end{equation}

For this metric, we can prove a couple of useful facts.

\begin{theorem}
For $U\in U(M)$ such that $U\not\in {\rm im}\,\varphi$, let $v\in \mathfrak{u}(M)$ be the principal logarithm of $U$. Let
\begin{equation}
v=v_T+v_N
\end{equation}
be the orthogonal decomposition of $v$, with a tangent component $v_T \in {\rm im}\, d\varphi$ and a normal component $v_N \in {\rm im}\, (d\varphi)^\perp$. Then,
\begin{enumerate}
\item $U_a=\exp(v_T)\in {\rm im}\, \varphi$.
\item $\Vert U -U_a\Vert\leq \Vert v_N\Vert$. 
\end{enumerate}
\end{theorem}

Therefore, for any normalized $\vert \psi\rangle$ with $\langle \psi\vert \psi\rangle=1$, we have
\begin{equation}
1\geq\vert \langle U \psi\vert U_a\psi\rangle\vert \geq 1-\frac{\Vert v_N\Vert^2}{2}.
\end{equation}

The proof is given by introducting a bi-invariant metric and reducing the issue to a problem in plane geometry thanks to Toponogov's comparison theorem \cite{CE08}. Later, with this theorem, we can give a recursive method to find a locally optimal approximation and show it converges.

\section{Prerrequisites and Notation}
\label{notation}
For the very basic notions in differential geometry, such as manifold, curve, tangent space, etc., the reader is referred to the books of Do Carmo \cite{DoCar92} or Sakai \cite{Sak96}.
\medskip

A Riemannian metric on a differentiable manifold $M$ is a correspondence which associates to each point $p$ on $M$ an inner product $\langle ~, ~\rangle_p$ on the tangent space $T_pM$ which varies differentiably in the sense that, for any pair of vector fields $X$ and $Y$ which are differentiable in a neighborhood $V$ of $M$, the function $\langle X,Y\rangle$ is differentiable on $V$. The metric with which a Riemannian manifold $M$ is endowed may come from a distance. Given two points $p,q \in M$, the distance $d(p,q)$ between them is defined to be the infimum of the lengths of all curves joining $p$ and $q$ which are piecewise differentiable. 
\medskip

Two fundamental concepts of Riemannian geometry are those of geodesic and curvature. Roughly speaking, a geodesic is a curve minimizing the distance between two nearby points. More precisely,  let $I$ be a closed interval in $\mathbb{R}$;
a parametrized curve $\gamma: I \to M$ is called a geodesic at $t_0$ if the covariant derivative $\frac{D}{dt}\left ( \frac{d\gamma}{dt}\right )$ vanishes at the point $t_0$ (see i.e. \cite{DoCar92}, Definition 2.1); if $\gamma$ is a geodesic at $t$ for all $t \in I$, then $\gamma$ is called a geodesic. If $[a,b] \subseteq I$ and $\gamma:I \to M$ is a geodesic, the restriction of $\gamma$ to $[a,b]$ is called a \emph{geodesic segment joining} $\gamma(a)$ to $\gamma (b)$. By abuse of language it is often referred to the image $\gamma (I)$ of a geodesic $\gamma$ as a geodesic.

A \emph{minimal geodesic} between $p$ and $q$ is the shortest one joining $p$ and $q$. It is easily seen that \emph{if} there exists a minimal geodesic $\gamma$ joining $p$ to $q$, then $d(p,q)$ equals the length $\ell(\gamma)$ of $\gamma$. This conditional \emph{if} holds under the hypothesis of completeness: a Riemannian manifold $M$ is said to be (geodesically) complete if for every $p \in M$ the exponential map $\mathrm{exp}_p$ is defined for the whole tangent space $T_p M$, i.e., if any geodesic $\gamma (t)$ starting from $p$ is defined for all $t \in \mathbb{R}$, and the statement is:

\begin{theorem}[Hopf-Rinow]\label{thm:hr}
Let $M$ be a Riemannian manifold and let $p \in M$. Then $M$ is geodesically complete if and only if it is complete as a metric space. Moreover, this implies that for any $q \in M$ there exists a minimizing geodesic $\gamma$ joining $p$ to $q$.
\end{theorem}

An important consequence of Theorem \ref{thm:hr} is the following, see \cite{Sak96}, Corollary 1.4 and Problem 1 of Chapter III:

\begin{corollary}
A $C^{\infty}$ manifold $M$ is compact if and only if any Riemannian metric on $M$ is complete.
\end{corollary}

On the other hand, the concept of curvature we will refer to is that of sectional curvature. According to Milnor \cite{Mil76} (p.~295), the sectional curvature of the tangential 2-plane spanned by some orthogonal unit vectors $u$ and $v$ can be described geometrically as the Gaussian curvature, at the point, of the surface swept out by all geodesics having a linear combination of $u$ and $v$ as tangent vector.
\medskip

We are interested in Riemannian manifolds with additional algebraic structure: Lie groups. A Lie group is a group $G$ with a differentiable structure such that the mapping $G\times G \to G$ given by $(x,y)\to xy^{-1}$, $x,y \in G$, is differentiable. It follows that \emph{translations from the left} $L_x$ resp. translations from the right $R_x$ given by $L_x: G \to G, L_x(y)=xy$ resp. $R_x: G \to G, R_x(y)=yx$ are diffeomorphisms. 

A Riemannian metric on $G$ is said to be \emph{left invariant}~resp.~\emph{right invariant} if for all $p, g \in G$ and for all $u,v \in T_pG$ it holds that 
$$
\langle u,v\rangle_p = \langle d(L_g)(u) , d(L_g)(v) \rangle_{L_g(p)}\ \ ~\mbox{resp.}~\ \ \langle u,v\rangle_p = \langle d(R_g)(u) , d(R_g)(v) \rangle_{R_g(p)}.
$$
A Riemannian metric is called \emph{bi-invariant} if it is both left and right invariant. Any compact Lie group can be endowed with a bi-invariant metric \cite[Exercise 7]{DoCar92}.

We also consider the Lie algebra $\mathcal{G}$ of $G$, which consists of the vectors in $T_e G$ with $e$ the neutral element of $G$ and with a well-known additional structure provided by a commutator (or Lie bracket) in the usual way.
\medskip

We will focus on the Lie group $U(M)$ as a differentiable manifold, for a positive integer $M$; its Lie algebra will be denoted by $\mathfrak{u}(M)$. Identity matrices of any size will be denoted by $Id$.
\medskip

\section{A bi-invariant Riemannian metric}\label{section:bi}
In this section, we endowe $U(M)$ with a Riemannian structure which will be useful later on.

For $u,v \in \mathfrak{u}(M)$, we define an inner product
\begin{equation}\label{eqn:skalar}
\langle u,v \rangle:=\frac{1}{2}\mathrm{tr}(u^{\dag} v + v^{\dag}u).
\end{equation}

This definition does actually correspond to a positive definite symmetric bilinear form: The bilinearity is an easy exercise, and moreover:
\begin{enumerate}
\item Since $u,v \in \mathfrak{u}(M)$, then $u^{\dag} = -u$, $v^{\dag}=-v$, and therefore
$$
\langle u,v \rangle:=\frac{1}{2}\mathrm{tr}(-u v - v u) = -\mathrm{tr}(uv).
$$
\item The symmetry of (\ref{eqn:skalar}) is clear, since $\langle u,v \rangle = -\mathrm{tr}(uv) = -\mathrm{tr}(vu) = \langle v,u \rangle$.
\item The positive definiteness follows from the fact that
$$
\langle u,u \rangle = \mathrm{tr}(u^{\dag}u)= ||u||^2\geq 0, 
$$
where $||u||=\sqrt{\sum_{i,j}|u_{ij}|^2}$ is the Frobenius norm of $u$, which is nonnegative and has all the required norm properties \cite{HJ85}.
\end{enumerate}

The metric defined above is Riemannian and bi-invariant. Bi-invariant metrics are useful for us because of the following.

\begin{theorem}[Milnor \cite{Mil76}]\label{thm:milnor}
Every compact Lie group admits a bi-invariant metric, which has nonnegative sectional curvature.
\end{theorem}

In the case of a bi-invariant metric, the sectional curvature admits an easier formula, see \cite[p.~323, Eqn.~(7.3)]{Mil76}:
$$
\kappa(u,v)=\frac{1}{4}\langle [u,v],[u,v] \rangle.
$$

Furthermore, in a Lie group admitting a bi-invariant metric, geodesic curves have an easy description: they coincide with the exponential map. More precisely, for $p \in U(M)$, a geodesic curve $\gamma: I \to U(M)$ such that $\gamma(0)=p$ and $\dot{\gamma}(0)=u$ is of the form
\begin{equation}\label{eqExpExp}
\gamma(t)=\exp{(up^{-1}t)}\cdot p
\end{equation}

In fact, for $p,q\in U(M)$, there exists a geodesic $\gamma$ joining $p$ and $q$ with $\gamma(0)=p$ such that
\begin{align}
\ell(\gamma([0,t]))=&\int_{0}^{t} \sqrt{\langle \dot{\gamma}(t),\dot{\gamma}(t)  \rangle} dt \nonumber\\
=& \int_0^{t} \sqrt{\langle u,u \rangle} dt = \int_0^{t} \norm{u}dt\label{eqlength} \\
=& \norm{u}t. \nonumber
\end{align}

A geodesic segment  $\gamma: [0,1] \to U(M)$ is called \emph{minimal} if it realizes a distance for any $t\in [0,1]$ , \emph{i.e.},
\begin{equation}\label{realizing}
d(\gamma(0),\gamma(t)) = \ell(\gamma([0,t])).
\end{equation}
By equation (\ref{eqExpExp}) a geodesic segment $\gamma: [0,1] \to U(M)$ joining $\gamma(0)=p$ and $\gamma(1)=q$ can be always obtained as
$$
\gamma(t)=\exp(vt)p,\quad \text{  with } v\text { such that } \exp(v)=q p^{-1}. 
$$
The concept of a minimal geodesic is related with the concept of principal logarithm in the following way
\begin{lemma}
\label{principalminimal}
Let $p,q\in U(M)$, let $w$ be a principal logarithm of $qp^{-1}$. Then, the geodesic segment
$$
\gamma:[0,1]\to U(M),\quad t\mapsto \gamma(t)=\exp(wt)p
$$
is a minimal geodesic segment joining $p$ and $q$.
\end{lemma}

Recall $iK$ is called a \emph{principal logarithm} of a unitary matrix $M\in U(M)$ if
$$
K^\dag=K,\quad \exp(iK)=M,\quad \text{ and  the  eigenvalues of $K$ are in }  (-\pi,\pi]. $$ 
There are efficient algorithms that can compute the principal logarithm of a unitary matrix \cite{Lor14}. We choose this definition of principal logartihm over the one for the interval $(-\pi,\pi)$ so that there is always a principal matrix logarithm, even for matrices with real negative eigenvalues (-1). Some properties, like infinite differentiability, are lost with this definition, but they are not used in our result.

\begin{proof}[Proof of Lemma \ref{principalminimal}]
The length of  a geodesic segment $\gamma(t)=\exp(vt)p$   with $\exp(v)=q p^{-1}$ is  (by equation (\ref{eqlength}))
$$
\ell(\gamma([0,1]))=\Vert vp\Vert=\Vert v\Vert.
$$ 
The distance $d(p,q)$ between $p$ and $q$ is the shortest length of the curves joining $p$ and $q$. In the case of a complete metric this shortest length is attained by a geodesic segment joining $p$ and $q$. Hence,  the geodesic segment $\gamma(t)=\exp(vt)p$ is  minimal  if and only if
$$
\Vert v\Vert=\min\{\Vert w\Vert\, :\, \exp(w)= q p^{-1}\}.
$$
But the equation $\exp(w)= q p^{-1}$  has the following family of solutions 
$$
w=U [\log(\lambda_i )\delta_{ij}] U^\dag
$$
with $pq^{-1}=U \Lambda U^\dag$, $U$ being a unitary matrix, $[\Lambda]_{ij}=(\lambda_i )\delta_{ij}$ being the diagonal matrix of eigenvalues of $pq^{-1}$, and $\log(\lambda_i)$ being any  logarithm of $\lambda_i$. Observe that
$$
\Vert w\Vert=\sqrt{\sum_{i=1}^M\vert \log(\lambda_i)\vert^2}.
$$
This expression only depends on the list of eigenvalues of $pq^{-1}$ and their logarithms. Since $pq^{-1}$ is unitary 
$$
\log(\lambda_i)=i(k_i+2\pi l_i),\quad \text{ with } k_i\in (-\pi,\pi], \text{ and } l_i\in\mathbb{Z}.
$$
Then
$$
\{\Vert w\Vert\, :\, \exp(w)= q p^{-1}\}=\left\{\sqrt{\sum_{i=1}^M\vert(k_i+2\pi l_i)\vert^2}\, :\, {l_i}\in\mathbb{Z}\right\},
$$
with a minimum when $l_i=0$ for $i=1,\cdots, M$, which corresponds to a principal logarithm.
\end{proof}

\section{An application of Toponogov's comparison theorem}\label{section:topo}
Riemannian manifolds whose curvature is bounded below may be investigated by applying Toponogov's comparison theorem. We first need to define triangles on the Riemannian manifold.

\begin{definition}
A \emph{geodesic triangle} $T=\Delta(p_1p_2p_3)$ of a Riemannian manifold $M$ is a set consisting of three segments of minimal geodesics, which are called the sides of $T$, say
$$
\gamma_1:[0,1] \to M, \ \gamma_2: [0,1] \to M, \ \ \mbox{and} \ \ \gamma_3: [0,1] \to M,
$$
such that $\gamma_i(1) = \gamma_{i+1}(0)$ for $i=1,2$, and $\gamma_3(1)=\gamma_1(0)$. The endpoints $p_1$, $p_2$ and $p_3$ are called the vertices of the triangle. The angle between the tangent vectors to $\gamma_{i-1}$ and $\gamma_{i+1}^{-1}$ at $p_i$ is called the angle of $T$ at $p_i$, and denoted by $\alpha_i=\angle(p_{i-1}p_ip_{i+1})$ or $\angle p_i$. The perimeter $\ell$ is defined as the sum $\ell(\gamma_1)+\ell(\gamma_2)+\ell(\gamma_3)$; if  we consider, in addition, that the two sides $\gamma_2, \gamma_3$ are minimal geodesics, and the side $\gamma_1$ is a geodesic segment, not necessarily minimal, with $\ell(\gamma_1)\leq \ell(\gamma_2)+\ell(\gamma_3)=d(p_1,p_3)+d(p_1,p_2)$, then the set is said to be a \emph{generalized geodesic triangle} (Figure \ref{geodesictriangle}).
\end{definition}

\begin{figure}[h]
\centering
\begin{tikzpicture}[scale=0.8]
    \draw [line width=1.5pt, fill=gray!2] (0,0) -- (60:4) -- (4,0) -- cycle;

    \coordinate[label=left:$p_2$]  (A) at (0,0);
    \coordinate[label=right:$p_3$] (B) at (4,0);
    \coordinate[label=above:$p_1$] (C) at (2,3.464);

    \coordinate[label=below:$\gamma_1$](c) at ($ (A)!.5!(B) $);
    \coordinate[label=left:$\gamma_3$] (b) at ($ (A)!.5!(C) $);
    \coordinate[label=right:$\gamma_2$](a) at ($ (B)!.5!(C) $);
    \draw[fill=green!30] (0,0) -- (0:0.75cm) arc (0:60:.75cm);
    \draw (0.45cm,0.25cm) node {$\alpha_2$};
    \begin{scope}[shift={(4cm,0cm)}]
        \draw[fill=green!30] (0,0) -- (-180:0.75cm) arc (180:120:0.75cm);
        \draw (150:0.5cm) node {$\alpha_3$};
    \end{scope}
    \begin{scope}[shift={(60:4)}]
        \draw[fill=green!30] (0,0) -- (-120:.75cm) arc (-120:-60:.75cm);
        \draw (-90:0.5cm) node {$\alpha_1$};
    \end{scope}
    \draw [line width=1.5pt] (A) -- (B) -- (C) -- cycle;
  \end{tikzpicture}
  \caption{Geodesic triangle $\Delta(p_1p_2p_3)$.\label{geodesictriangle}}
\end{figure}
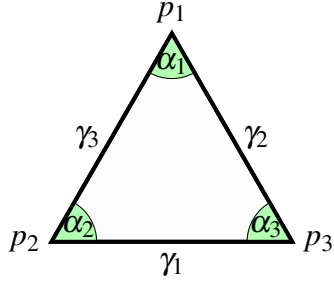

Let us set the part of Toponogov's comparison theorem we are interested in, cf.~\cite[Theorem 4.2 in Chapter IV]{Sak96}:

\begin{theorem}[Toponogov]\label{TCT}
Let $M$ be a complete Riemannian manifold whose sectional curvatures satisfy $\kappa \geq \delta$ everywhere for some constant $\delta$. 
Denote by $M_{\delta}^2$ the 2-dimensional complete simply connected Riemannian manifold of constant curvature $\delta$. Consider a generalized geodesic triangle $\Delta(p_1p_2p_3)$ such that $\gamma_2, \gamma_3$ are minimal and $\ell(\gamma_1) \leq \pi/\sqrt{\delta}$. Then the perimeter $\ell \leq 2\pi/\sqrt{\delta}$ and there exists a geodesic triangle $\Delta(\tilde{p}_1\tilde{p}_2\tilde{p}_3)$ in $M_{\delta}^2$ with the same side lengths $\ell(\tilde{\gamma_i})=\ell(\tilde{\gamma_i})$, for $i=1,2,3$ and satisfying $\alpha_2 \geq \tilde{\alpha}_2$ and $\alpha_3 \geq \tilde{\alpha}_3$.
\end{theorem}

\noindent \emph{Remark.} It is to assume that $\pi/\sqrt{\delta} = +\infty$ when $\delta \leq 0$ in the theorem above.
\medskip

Theorem \ref{TCT} allows us to compare triangles of $U(M)$ with triangles in $\mathbb{R}^2$ by setting $\delta=0$, since in this case $M_{0}^2=\mathbb{R}^2$.
In our situation it is  $p_1=U$,  $p_3=Id$  and $p_2=U_a=\exp(v_T)$  is our approximation matrix in $\mathrm{im}(\varphi)$ (recall that $v_T$ is the tangential component of  the principal logarithm $v$  of $U$).  
Set $
 \gamma_1(t)=U_a\exp(-v_Tt),
$
$
 \gamma_2(t)=\exp(vt)
$ and
$\gamma_3(t)$ a minimal geodesic segment joining $U$ with $U_a$.
Then  $\gamma_2$ and  $\gamma_3$ are minimal geodesic segments and $\gamma_1$ is a geodesic segment (not necsarily minimal). 
Set $\ell_1=\ell(\gamma_1([0,1]))=\norm{v_T}$, $\ell_2=\ell(\gamma_2([0,1]))=d(U,Id)$ and $\ell_3=\ell(\gamma_3([0,1]))=d(U,U_a)$. 
Observe that
$$
\ell_1=\norm{v_T}\leq \norm{v}+d(U,U_a)=\ell_2+\ell_3.
$$
Hence, 
by Theorem \ref{TCT}, there exists a geodesic triangle in $\mathbb{R}^2$ with sides $\vec{\ell_1}, \vec{\ell_2}$ and $\vec{\ell_3}$ of lengths $\ell_1, \ell_2$ resp. $\ell_3$, such that $\alpha_2 \geq \tilde{\alpha}_2$ and $\alpha_3 \geq \tilde{\alpha}_3$. 
We want to estimate the distance $\ell_3$. First of all notice that $\vec{\ell_3}=\vec{\ell_2}- \vec{\ell_1}$, hence the law of cosinus implies that
$$
\norm{\vec{\ell_3}}^2=\norm{\vec{\ell_2}- \vec{\ell_1}}^2 = \norm{\vec{\ell_1}}^2+\norm{\vec{\ell_2}}^2-2\norm{\vec{\ell_1}}\norm{\vec{\ell_2}} \cos{\angle(\vec{\ell_1},\vec{\ell_2})}
$$
and
\begin{equation}\label{eq:4.1}
\ell_3^2=\ell_1^2+\ell_2^2-2\ell_1\ell_2\cos{\tilde{\alpha}_3}.
\end{equation}
Since $\cos{{\alpha}_3}=\frac{\langle v,v_T\rangle}{\norm{v}\norm{v_T}}=\frac{\norm{v_T}}{\norm{v}}\geq 0$, it follows that  $\tilde{\alpha}_3 \leq \alpha_3\leq \frac{\pi}{2}$ which implies $-\cos{\tilde{\alpha}_3} \leq -\cos{\alpha_3}$ and, so,  
\begin{equation}\label{eq:eles}
\begin{aligned}
\ell_3^2\leq &\ell_1^2+\ell_2^2-2\ell_1\ell_2\cos{\alpha_3}=\norm{v_T}^2+\norm{v}^2-2\norm{v_T}\norm{v}\cos{\alpha_3}\\
=&\norm{v}^2+\norm{v_T}^2-2\langle v, v_T \rangle=\norm{v-v_T}^2=\norm{v_N}^2.
\end{aligned}
\end{equation}

Figure \ref{graphrep} shows a graphical representation of our scenario.
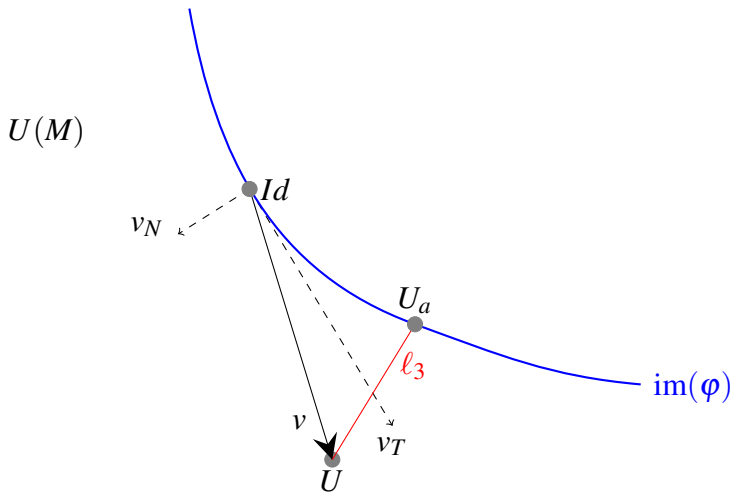
\begin{figure}[h]
\centering
\begin{tikzpicture}
\draw  (0.8,2.6) node [above, right] {$Id$};
\draw  (3,.8) node [above] {$U_a$};

\draw (-1.9,3) node [above] {$U(M)$};

\draw[dashed, ->] (0.8,2.6) -- (2.7,-0.55);
\draw[dashed, ->] (0.8,2.6) -- (-0.15,2);
\draw (2.7,-0.55) node [below] {$v_T$};
\draw (-0.2,2.1) node [below, left] {$v_N$};
\draw[red] (3.3,0.25) node [below, left] {$\ell_3$};
\draw[fill,gray] (1.9,-1) circle [radius=.1];
\draw[decoration={markings,mark=at position 1 with
    {\arrow[scale=3,>=stealth]{>}}},postaction={decorate}]  (0.8,2.6) -- (1.9,-1);
\draw (1.7,-0.5) node [below,left] {$v$};
\draw (1.9,-1) node [below] {$U$};
\draw[red] (3,.8) -- (1.9,-1);
\draw[blue] [thick] (0,5) to [out=-80, in=160] (3,.8) to [out=-20, in=175] (6,0);
\draw[blue] (6.7,-0.4) node [above] {$\mathrm{im}(\varphi)$};
\draw [fill,gray] (3,.8) circle [radius=.1];
\draw [fill,gray] (0.8,2.6) circle [radius=.1];
\end{tikzpicture}
\caption{Submanifold $\mathrm{im}(\varphi)$ in $U(M)$, matrix $U$ and approximation $U_a$.\label{graphrep}}
\end{figure}


Now we want to see the bi-invariant metric in $\mathfrak{u}(M)$ defined in Section \ref{section:bi} as a metric in $U(M)$.
First observe that the Riemannian manifold $(U(M),\langle , \rangle)$ is a Riemannian subvariety of the manifold $\mathcal{M}_n(\mathbb{C})$ of the complex $n\times n$-matrices endowed with the Euclidean inner product. If we write $d_{U(M)}$ for the distance we use on $U(M)$, and $d_{\mathcal{M}_n}$ for the one on 
$\mathcal{M}_n(\mathbb{C})$, then
$$
\ell_3\geq d_{U(M)} (U_a,U) \geq d_{\mathcal{M}_n} (U_a,U).
$$
The inequality holds since ``distance'' between two points is the infimum of the length of any two curves joining the points. Now it is easy to find a minimal geodesic for $\mathcal{M}_n(\mathbb{C})$ joining $U_a$ and $U$, namely the segment $\gamma_4(t)=U_a+(U-U_a)t$ for $t\in [0,1]$. Since $\dot{\gamma_4}(t)=U-U_a$, we find that $ d_{\mathcal{M}_n} (U_a,U)=\norm{\dot{\gamma_4}(t)}\cdot 1=\norm{U-U_a}$ and, therefore, 
$$
\ell_3 \geq \norm{U-U_a}.
$$

This inequality together with (\ref{eq:eles}) yields 
$$
\norm{U-U_a}^2 \leq \ell_3^2 \leq \norm{v_N}^2,
$$
and we obtain, finally,
$$
\norm{U-U_a} \leq \norm{v_N}.
$$

\section{Random unitary matrices and a basis for the image}
\label{tricks}
Before looking into the iterative process that gives the locally optimal approximation to any desired $U\not\in {\rm im}\varphi$, we need to consider a few useful tricks.

While in the previous section we have considered the identity matrix as our starting point in the image group, the results hold for any $U_0 \in {\rm im}\varphi$. The identity has some advantages: it is always in the image, for any values of $n$ and $m$, and, in terms of computation, it is trivial to generate the $M\times M$ identity matrix. 

However, in general, the landscape of the image group is unknown and numerical experiments show there are multiple local minima. In a 2D space, we can picture the image as an irregular profile with mountains and valleys. For any given starting point, the iterative process ends in a local minimum, but, as we do not know how many minima exist, we need a series of random starting points to sample multiple approximations, each the closest to the intended $U$ in its local neighbourhood.

In order to generate a random unitary $U_r\in {\rm im}\varphi $, we choose a random $S_r \in U(m)$ uniformly from all the possible matrices and compute $U_r=\varphi(S_r)$. The random unitaries in $U(m)$ can be generated from samples of a normal distribution \cite{Mez07,Tot08} and the photonic homomorphism is known \cite{SGL04,Sch04}.

Finally, we need a way to project the logarithm of any $U\in U(M)$ to the image algebra ${\rm im}d\varphi$. The method is similar to the generation of random unitaries. We start by working on the algebra corresponding to the scattering matrices, $\mathfrak{u}(m)$, where we can write down a known basis and, using the differential map $d\varphi$, explicitly given in \cite{GGM18}, we can obtain a basis for ${\rm im}d\varphi$, which is a subalgebra of $\mathfrak{u}(M)$. This basis is not necessarily orthonormal, but it can be orthogonalized and normalized. The details of the whole procedure can be found in \cite{GGM19}. This basis, together with the inner product of Eq. (\ref{eqn:skalar}) are enough to obtain the desired projection on the image algebra. 

\section{An iterative process for the approximation}
\label{algo}
In Section \ref{section:topo} we have constructed an approximation matrix $U_1:=U_a \in \mathrm{im}(\varphi)$, starting from $U=\exp{(v)}$ as 
$$
U_1=\exp{((\log{U})_T)}, \ \  \ \ \mbox{with} \ \ \ \ \norm{U-U_1} \leq \norm{(\log{U})_N}.
$$
Now, we can repeat this by taking a new approximation $U_2\in \mathrm{im}(\varphi)$ by considering a geodesic triangle of vertices $Id$, $U_1^{-1}U$ and $U_2$, where $U_1^{-1}U=\exp{(-v_T)}\exp{v}$ and
$$
U_2=U_1\exp{((\log{U_1^{-1}U})_T)} \ \ \mbox{with} \ \  \norm{U-U_2} \leq \norm{(\log{U^{-1}_1U})_N}.
$$
Iteratively,
$$\left\{\begin{aligned}
  U_0=& {\rm Id}\\
  U_n=& U_{n-1}\exp((\log(U_{n-1}^{-1}U)_T))
\end{aligned}\right. $$
with
$$
\norm{U-U_n}\leq \norm{(\log(U_{n-1}^{-1}U)_N}.
$$

\subsection{Convergence}
The method described above converges. Consider the sequence $\{U_i\}$ with $i\geq 0$ and $U_0=Id$. For the first step we know that $d(U_1,U)=\Vert v^1\Vert$, $d(U_2,U)\leq \Vert v^1_N\Vert$, and $d(U_1,U_2)\leq \Vert v^1_T\Vert$. In general, we have:
\begin{enumerate}
\item $d(U_i,U)=\Vert v^i\Vert$, 
\item $d(U_{i+1},U)\leq \Vert v^i_N\Vert$,  
\item $d(U_i,U_{i+1})\leq \Vert v^i_T\Vert$.
\end{enumerate}

\begin{proposition}\label{prop:ineq}
We have that
\begin{equation}\label{eq:ineq}
d(U_{i+1},U) \leq d(U_i, U).
\end{equation}
Furthermore, the equality holds if and only if $U_i=U_{i+1}$.
\end{proposition}

\begin{proof}
By a successive application of the inequalities (2) and (1) above we obtain
$$
d(U_{i+1},U) \overset{(2)}{\leq} \Vert v^i_N\Vert   \leq  \Vert v^i \Vert \overset{(1)}{=} d(U_i, U).
$$
Now, if the equality (\ref{eq:ineq}) holds, then $\Vert v^i_N\Vert   =  \Vert v^i \Vert $ and so $\Vert v^i_T\Vert  =0 $; inequality (3) above allows us to conclude that $U_i=U_{i+1}$. The converse is trivial.
\end{proof}

Let us define 
$$
d:\mathbb{N} \to \mathbb{R}, \ \ i \mapsto d_i:=d(U_i,U).
$$

\begin{proposition}
The sequence $\{d_i\}$ is convergent.
\end{proposition}

\begin{proof}
Assume that $U_i \neq U_{i+1}$ for every $i$ (otherwise, there would exist $n \in \mathbb{N}$ such that $U_n=U_{n+1}$, and therefore $d_n=d_{n+1}=\cdots$ and the sequence converges to $d_n$). By Proposition \ref{prop:ineq} the sequence $\{d_i\}$ is decreasing; in particular $d_i <d_1$ for all $i$. This together with the fact that $d_i \geq 0$ for all $i$ (i.e., the sequence $\{d_i\}$ is bounded) implies the convergence.
\end{proof}

In fact, the approximation given by the described method is the best possible one in the following sense. Since $\{d_i\}$ is a convergent sequence, then it is a Cauchy-sequence. This means that for every $\epsilon >0$ there exists $n_{\epsilon} \in \mathbb{N}$ with
$$
0<d_i-d_{i+1} < \epsilon, \ \ \mbox{for all } \ i > n_{\epsilon}.
$$
Since, by (1) and (2) we have that $d_i=\Vert v^i\Vert$ and $d_{i+1} \leq \Vert v^i_N\Vert$, it is easily deduced that
\begin{equation}\label{eq:4}
0\leq \Vert v^i\Vert - \Vert v^i_T\Vert \leq d_i - d_{i+1} < \epsilon.
\end{equation}

We first observe the following inequality:

\begin{lemma}\label{lemma:lemma}
$$
\Vert v^i -v^i_N\Vert^2 \leq (\Vert v^i\Vert-\Vert v^i_N\Vert)\cdot 2 \cdot d_1.
$$
\end{lemma}

\begin{proof}
\begin{align*}
\Vert v^i -v^i_N\Vert^2 = &\Vert v^i\Vert^2 + \Vert v^i_N\Vert^2-2\langle v^i,v^i_N \rangle \\
=& \Vert v^i\Vert^2 + \Vert v^i_N\Vert^2-2\langle v^i_T+v^i_N,v^i_N \rangle \\
=&\Vert v^i\Vert^2 + \Vert v^i_N\Vert^2-2\Vert v^i_N \Vert^2 = \Vert v^i \Vert^2 - \Vert v^i_N \Vert^2 \\
=& (\Vert v^i \Vert - \Vert v^i_N \Vert )(\Vert v^i \Vert + \Vert v^i_N \Vert ).
\end{align*}
Since $\Vert v^i_N \Vert \leq \Vert v^i \Vert$, it holds that
\begin{align*}
(\Vert v^i \Vert - \Vert v^i_N \Vert )(\Vert v^i \Vert + \Vert v^i_N \Vert ) \leq &(\Vert v^i \Vert - \Vert v^i_N \Vert )2 \Vert v^i \Vert \\
\overset{(1)}{=} & (\Vert v^i \Vert - \Vert v^i_N \Vert ) 2d_i.
\end{align*}
Proposition \ref{prop:ineq} implies that $d_i \leq d_1$, hence we get
$$
\Vert v^i -v^i_N\Vert^2 \leq (\Vert v^i \Vert - \Vert v^i_N \Vert ) 2d_1.
$$
\end{proof}

Now, inequality (\ref{eq:4}) together with Lemma \ref{lemma:lemma} imply that
$$
0 \leq \Vert v^i - v^i_N\Vert ^2 \leq 2\epsilon d_1.
$$
On the other hand, $\Vert v^i-v^i_N \Vert^2 = \Vert v_T \Vert^2$ and, together with inequality (3), this yields
$$
d(U_i,U_{i+1}) \leq \Vert v^i - v^i_N \Vert^2 \leq 2 \epsilon d_1.
$$

Therefore the sequence $\{U_i\}$ is a Cauchy-sequence itself. This allows us to apply Theorem \ref{thm:hr} (Hopf-Rinow) in the complete metric space $(M,g)$, and so the sequence $\{U_i\}$ converges to a matrix $\tilde{U} \in M$, and $\Vert v^T_{\infty} \Vert =0$. Hence the geodesic joining $\tilde{U}$ with $U$ is normal.

\medskip

\subsection{Application example}
We can see an example of the procedure with the Quantum Fourier Transform 
\begin{equation}
Q\!F\!T\ket{x}=\frac{1}{\sqrt{M}}\sum_{y=0}^{M-1}e^{\frac{i2\pi xy}{M}}\ket{y}. 
\end{equation}

We choose the QFT not only for being a useful transformation in quantum information and quantum optics, but also because it can be considered one of the most difficult transformations for linear optics. The QFT is impossible to achieve with linear systems with $m>1$ and more than one photon and has a strong structure.

We start from the QFT matrix for $M=3$ ($n=m=2$)
$$\begin{aligned}
\label{QFT3}
  U=&\frac{1}{\sqrt{3}}\left(
\begin{array}{ccc}
 1 & 1 & 1 \\
 1 & e^{-i\frac{2\pi}{3}} & e^{-i\frac{4\pi}{3}}\\
 1 & e^{-i\frac{4\pi}{3}} & e^{-i\frac{8\pi}{3}} \\
\end{array}
\right)\\
 \approx&\left(\begin{matrix}0.57735 & 0.57735 & 0.57735\\0.57735 & -0.28868 - 0.5 i & -0.28868 + 0.5 i\\0.57735 & -0.28868 + 0.5 i & -0.28868 - 0.5 i\end{matrix}\right).
\end{aligned}
$$

For the implementation of $U$, we consider the states in the ordered basis $\{\ket{20},\ket{02},\ket{11}\}$. All the results are given with 5 significant digits for matrix entries and 10 significant digits for distances, rounding the imaginary or real parts to zero when they are much smaller than the surrounding terms. 

The starting point is the identity matrix, for which $\norm{U-{\rm Id}}=2.449489743$. In the first step of our procedure, we take the projection of the principal logarithm of $U$ into the image algebra:
\begin{equation}
\nonumber
\log(U)=\left(\begin{matrix}- 0.6639 i & 0.9069 i & 0.9069 i\\0.9069 i & - 2.0242 i & - 0.45345 i\\0.9069 i & - 0.45345 i & - 2.0242 i\end{matrix}\right)
\end{equation}
and
\begin{equation}
\nonumber
\log(U)_T=
\left(\begin{matrix}- 0.89062 i & 0 & 0.22672 i\\0 & - 2.251 i & 0.22672 i\\0.22672 i & 0.22672 i & - 1.5708 i\end{matrix}\right)
\end{equation}
so that
\begin{equation}
\nonumber
U_1=\left(\begin{matrix}0.61786 - 0.75486 i & 0.024514 i & 0.20595 + 0.073541 i\\0.024514 i & -0.61786 - 0.75486 i & 0.20595 - 0.073541 i\\0.20595 + 0.073541 i & 0.20595 - 
0.073541 i & - 0.95097 i\end{matrix}\right)
\end{equation}
with $\norm{U-U_1}=1.770101749$.

After 10 steps, we have:
\begin{equation}
\nonumber
U_{10}=\left(\begin{matrix}0.86432 - 0.50294 i & 2.889 \cdot 10^{-6} i & 0.0020837 + 0.0011983 i\\2.889 \cdot 10^{-6} i & -0.86432 - 0.50294 i & 0.0020837 - 0.0011983 i\\0.002
0837 + 0.0011983 i & 0.0020837 - 0.0011983 i & - 0.99999 i\end{matrix}\right)
\end{equation}
with $\norm{U-U_{10}}=1.732054756$.
Further iterations only produce marginal improvements in the distance to the target matrix beyond the fifth decimal place.
 
We stop after 10 additional steps and use as an approximation 
\begin{equation}
\nonumber
U_{20}=\left(\begin{matrix}0.86601 - 0.50003 i & 0 & 0\\0 & -0.86601 - 0.50003 i & 0\\0 i & 0 & - 1.0 i\end{matrix}\right)
\end{equation}
with $\norm{U-U_{20}}=1.732050808$.

\subsection{Random initial matrices}
In fact, the results hold for any arbitrary $U_0\in {\rm im}\varphi$. Now the procedure becomes computationally more involved, as we need an explicit evaluation of $U_r=\varphi(S_r)$ for random $S_r$ matrices, with a complexity which grows combinatorily in $n$ and $m$. However, as we see in this second example, we can explore different local optima. 

For the QFT matrix in Eq. (\ref{QFT3}), if we start at random points, the approximation gravitates towards three solutions. Two of them are at the same distance from the QFT matrix than the approximation we found starting with the identity matrix, $1.7320$, with
\begin{equation}
\nonumber
U_a^1=\left(\begin{matrix}0.86602 - 0.5 i & 0 & 0\\0 & -0.86602 - 0.5 i & 0\\0 & 0 & - 1.0 i\end{matrix}\right)
\end{equation}
and
\begin{equation}
\nonumber
U_a^2=\left(\begin{matrix}0 & 0.86602 + 0.5 i & 0\\0.86602 + 0.5 i & 0 & 0\\0 & 0 & -0.86602 - 0.5 i\end{matrix}\right).
\end{equation}
There is a third solution with a distance 0.85675 to the QFT for the matrix:
\begin{equation}
\nonumber
U_a^3=\left(\begin{matrix}0.43301 + 0.25 i & 0.43301 - 0.25 i & 0.70711\\0.43301 - 0.25 i & - 0.5 i & -0.35355 + 0.61237 i\\0.70711 & -0.35355 + 0.61237 i & 0\end{matrix}\right).
\end{equation}
These solutions seem to be found with equal probability. For a run of 1000 experiments we found $U_a^1$ 311 times, $U_a^2$ 374 times and $U_a^3$ 315 times. Further experiments showed a similar behaviour. 

The best approximation to the $3\times 3$ QFT matrix $U_a^3$ comes from a scattering matrix:
\begin{equation}
\nonumber
S_a^3=\left(\begin{matrix}0.68301 + 0.18301 i & 0.68301 - 0.18301 i\\0.68301 - 0.18301 i & -0.5 + 0.5 i\end{matrix}\right)
\end{equation}
which is, up to a $\frac{5\pi}{12}$ global phase, a balanced two input beam splitter with a scattering matrix
\begin{equation}
\nonumber
S_{BS}=\frac{1}{\sqrt{2}}\left(\begin{matrix} 1 & i \\ i & 1\end{matrix}\right).
\end{equation}
preceded by a $\frac{-\pi}{3}$ phase shifter in the first port and followed by a $\frac{\pi}{3}$ in the second port.

Depending on the starting local point, there are different local optima. Sometimes, two different approximation matrices in the image group give the same distance to the target unitary. 

\section{Summary, recommendations and future improvements}
\label{summary}
We have given an iterative method that finds a linear optical setup that approximates any arbitrary quantum evolution for $n$ photons in $m$ modes. The approximation is optimal in the local neighbourhood of the initial guess. Once we have the closest unitary that can be implemented, $\widetilde{U}$, we can use a previous method \cite{GGM19} to obtain a scattering matrix $\widetilde{S}$ that gives the approximated evolution and there are multiple algorithms that give the physical setup corresponding to $\widetilde{S}$ using only beam splitters and phase shifters \cite{RZB94,CHM16}.

The proposed algorithm is based on results from differential geometry, in particular, Toponogov's theorem. They show the method will converge and find local optima. 

There are a few practical details worth mentioning. First, numerically, we find that if we try the method on an evolution which is possible to obtain from a linear optics system, $U\in {\rm im}\varphi$, sometimes it will converge to a matrix close the actual solution, but, depending on the local landscape, it might fall into a variety of different matrices all at a similar large distance. In practical applications, the recommendation would be, first, check whether an exact implementation exists (with our previous algorithm \cite{GGM19}) and, if there is none, look for an approximation.

There are also some open problems. From the structure of the involved groups and algebras, it is not clear how many local optima exist for any given transformation $\varphi$. We have proposed a randomized way of exploring the state space of the potential approximation matrices, and, for the limited dimensions that can be numerically explored, it seems to work well, but there is no guarantee the method finds a global minimum. Any further knowledge of the structure of the image group would help in the search for a global minimum or, at least, in finding a probabilistic bound on the optimal approximation. Additionally, a lower bound on $\norm{U-U_a}$, as opposed to our upper norm, would help to determine the global optimum.

Even in its present form, the proposed algorithm, when combined with previous results, can assist in the design of quantum optical operations and has applications to quantum information and quantum optics experiment design.

\newcommand{\noopsort}[1]{} \newcommand{\printfirst}[2]{#1}
  \newcommand{\singleletter}[1]{#1} \newcommand{\switchargs}[2]{#2#1}

\end{document}